\newif\iffinal

\iffinal

\documentclass[journal,twoside,web]{ieeecolor}
\usepackage{lcsys}

\else

\documentclass[letterpaper, 10 pt, conference]{ieeeconf}
\IEEEoverridecommandlockouts                             
\fi

\title{
\iffinal 
\else
\LARGE \bf
\fi
Almost Global Asymptotic Trajectory Tracking 
\iffinal \\ \fi
for Fully-Actuated Mechanical Systems on Homogeneous Riemannian Manifolds
}

\author{Jake Welde and Vijay Kumar%
\thanks{
	J. Welde and V. Kumar are with the GRASP Laboratory at the University of Pennsylvania.
	We sincerely thank Dr. Aradhana Nayak and Prof. Ravi N. Banavar, for helpful discussions on their work \cite{Nayak2019} during this work's development; Prof. Matthew D. Kvalheim, for 
	his help identifying a missing assumption in our theorem;
	and Prof. Jean Gallier, for helpful conversations.
	 We gratefully acknowledge the support of Qualcomm Research and the NSF Graduate Research Fellowship Program.
}
}

\usepackage{cite}

\iffinal
\usepackage{generic}
\fi

\usepackage{quiver}
\usepackage{amsmath,amssymb}
\usepackage{bbm}
\usepackage{dsfont}
\usepackage{centernot}

\newcommand{\pairing}[2]{\big\langle{\hspace{1pt}}#1{\hspace{1.5pt};\hspace{1.5pt}}#2{\hspace{1pt}}\big\rangle}

\DeclareMathOperator{\dist}{dist}

\newcommand{\diff}{\operatorname{d}\hspace{-1pt}}
\DeclareMathOperator{\Diff}{D}

\DeclareMathOperator{\spn}{span}

\DeclareMathOperator{\id}{id}

\DeclareMathOperator{\tr}{tr}
\DeclareMathOperator{\Ad}{Ad}
\DeclareMathOperator{\ad}{ad}

\usepackage{amsthm}
\usepackage{thmtools}
\renewcommand\thmcontinues[1]{{\normalfont continued}}

\declaretheoremstyle[notefont=\normalfont\itshape,bodyfont=\normalfont]{normaltext}

\newtheorem{theorem}{Theorem}
\newtheorem{corollary}{Corollary}
\declaretheorem[name=Definition,style=normaltext]{definition}
\declaretheorem[name=Remark,style=normaltext]{remark}
\declaretheorem[name=Example,style=normaltext]{example}
\newtheorem{proposition}{Proposition}
\newtheorem{lemma}{Lemma}

\usepackage{balance}
\usepackage{forloop}

\usepackage{graphicx,stackengine,scalerel}

\usepackage{enumerate}%

\makeatletter
\let\NAT@parse\undefined
\makeatother
\usepackage[hidelinks]{hyperref}

\begin{document}

\maketitle
\thispagestyle{empty}
\pagestyle{empty}

\begin{abstract}
	In this work, we address the design of tracking controllers that drive a mechanical system's state asymptotically towards a reference trajectory. Motivated by aerospace and robotics applications, we consider fully-actuated systems evolving on the broad class of homogeneous spaces (encompassing all vector spaces, Lie groups, and spheres of any finite dimension). In this setting, 	the transitive action of a Lie group on the configuration manifold enables an intrinsic description of the tracking error as an element of the state space, even in the absence of a group structure on the configuration manifold itself (\textit{e.g.}, for $\mathbb{S}^2$). Such an error state facilitates the design of a generalized control policy depending smoothly on state and time, which drives the geometric tracking error to a designated origin from almost every initial condition, thereby guaranteeing almost global convergence to the reference trajectory. Moreover, the proposed controller simplifies elegantly when specialized to a Lie group or the $n$-sphere. In summary, we propose a unified, intrinsic controller guaranteeing almost global asymptotic trajectory tracking for fully-actuated mechanical systems evolving on a broad class of manifolds. We apply the method to an axisymmetric satellite and an omnidirectional aerial robot.
\end{abstract}

\iffinal
\fi

\section{Introduction}

\iffinal
\IEEEPARstart{A}{n}
\else
An
\fi
efficient, effective approach to the control of mechanical systems 
is to synthesize an overall controller via the composition of a low-rate ``open-loop'' planner (to select a {reference trajectory}) with a high-rate ``closed-loop'' tracking controller (to drive the system towards the reference trajectory). Such layered control architectures are found across an extremely diverse array of engineered and biological systems, demonstrating the effectiveness of the paradigm \cite{Matni2024}.

For some mechanical systems, it is possible to reduce the \textit{``tracking problem''} to the easier \textit{``regulation problem''} (\textit{i.e.}, asymptotically stabilizing an equilibrium) for the same system, using a state-valued {``tracking error''} between the reference and actual states that ``vanishes'' only when the two are equal (\textit{e.g.}, 
$({q - q_d},{\dot q - \dot q_d})$ for a mechanical system on ${\mathbb{R}^n}$). Then, feedback is designed to drive the error towards the origin, \textit{i.e.}, a \textit{constant} (equilibrium) trajectory.
Such a reduction is substantially easier for fully-actuated (\textit{vs.} underactuated) systems, since arbitrary forces can be exerted, both to compensate for time variation in the reference
and to inject suitable artificial potential and dissipation terms to 
almost globally asymptotically stabilize 
the error 
\cite{Koditschek1989}.

Clearly, a smooth tracking controller suitable for operation on the entire state space (in general, the tangent bundle of a non-Euclidean manifold) \textit{cannot} rely on a local, coordinate-based tracking error.
When the configuration manifold is a Lie group, \cite{Maithripala2006} shows that the group structure furnishes an intrinsic, globally-defined 
tracking error on the state space, via the action of the group on itself.
For this reason, \cite{Maithripala2006}  remarks that \textit{``the tracking problem on a Lie group is more closely related to tracking on $\mathbb{R}^n$ than it is to the general Riemannian case, for which the group operation is lacking.''}

Tracking via error regulation on manifolds that may \textit{not} be Lie groups is proposed
in \cite{Nayak2019}. The authors
present a theorem ensuring almost global tracking for fully-actuated systems on compact Riemannian manifolds, assuming the control inputs 
can be chosen to solve a 
``rank one'' overdetermined
linear system at each state and time. 
Solutions exist in some special cases (\textit{e.g.}, on $\mathbb{S}^2$ with carefully chosen potential, dissipation, and kinetic energy), although they are non-unique. In other cases, \textit{no} solutions exist, and the question of \textit{when} solutions exist in general is left open.
Also, the inputs may not depend continuously on state and time, a desirable property that justifies seeking only \textit{almost} global asymptotic stability
\cite{Bernuau2013}.

\begin{figure}[t]
	\vspace{4pt}
	\begin{center}
	\hspace{2pt}
	\foreach \n in {5,53,101,149}{%
		\includegraphics[width=.23\columnwidth,trim={15.7cm 5cm 15.7cm 5cm}, clip]%
		{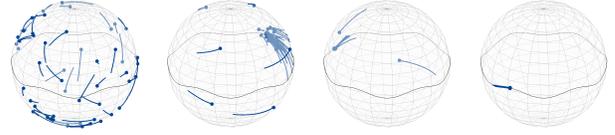}%
		\hfill%
	}
	\hspace{2pt}
\end{center}
	\vspace{-10pt}
	\caption{
		The proposed controller, applied to a mechanical system on $\mathbb{S}^2$ (\textit{e.g.}, the axisymmetric satellite). We show 
		snapshots of parallel rollouts from a random sampling of initial states (configuration and velocity) in $T\mathbb{S}^2$. All sampled rollouts converge to the reference trajectory. 
} 
	\vspace{-14pt}
	\label{sphere_tracking_figure}
\end{figure}

In view of these observations, and in pursuit of a systematic approach to tracking control for  fully-actuated mechanical systems in yet a broader setting, a question emerges:
\begingroup
\addtolength\leftmargini{.05in}
\begin{quote}
	\begin{center}
		\textit{On which class of manifolds can
			the tracking problem be reduced to the regulation problem?}
	\end{center}
\end{quote}
We pose this question in a global or almost global sense; if purely \textit{local} convergence suffices, the answer is trivial (\textit{i.e.}, ``all of them''), since smooth manifolds are locally Euclidean.

\endgroup

In this work, we consider mechanical systems evolving on 
\textit{homogeneous spaces}, a class of manifolds that includes all Lie groups but is more general.
We obtain globally valid ``error dynamics''
by extending an intrinsic error used in tracking for \textit{kinematic} systems in \cite{Hampsey2022} to the mechanical setting, exploiting the transitive action of a Lie group on the configuration manifold even when it lacks a group structure of its own. Ultimately, we contribute a systematic approach to tracking control with formal guarantees for a broad class of systems
that includes 
many space, aerial, and underwater robots.

\section{Mathematical Background}

For any map ${(x,y) \mapsto f(x,y)}$, we define the map ${f_x : y \mapsto f(x,y)}$ (resp. ${f^y : x \mapsto f(x,y)}$) by holding the first (resp. second) argument constant.
For a smooth map ${f : M \to N}$ between manifolds, its derivative is denoted ${\diff f : TM \to TN}$, while the dual of its derivative is the unique map ${\diff f^* : T^*N \to T^*M}$ such that ${\pairing{\omega_n}{\diff f(v_m)}= \pairing{\diff f^*(\omega_n)}{v_m}}$ for any vector ${v_m \in TM}$ and any covector ${\omega_n \in T_{f(m)}^*N}$.
The derivative of a function ${V : M \to \mathbb{R}}$ induces a covector field ${\diff V :M \to T^*M}$. We assume all given curves are smooth.

\subsection{Homogeneous Riemannian Manifolds}

\begin{definition}[see {\cite[Ch. 3]{RiemannianLee}}]
	\ 
	A 
	\textit{homogeneous Riemannian manifold} $(Q,\Phi,\kappa)$ is a smooth manifold $Q$ equipped with:
	\begin{enumerate}[i.]
		\item 	 a transitive left action $\Phi : G \times Q \to Q$ of a Lie group $G$, \textit{i.e.},
		${\forall \  q_1,q_2 \in Q}$, ${\exists \ g \in G \textrm{ s.t. } \Phi(g,q_1) = q_2}$, 
		and 
		\item a $\Phi$-invariant Riemannian metric $\kappa$, \textit{i.e.}, 
		${\forall \ g \in G}$ and ${v_q, w_q \in TQ}, \
			\kappa\big({v_q},{w_q}\big)
			=  
			\kappa\big({\diff\Phi_g (v_q),\diff \Phi_g (\omega_q)}\big)
		 $.
	\end{enumerate}
	The \textit{stabilizer} $G_q$ of ${q \in Q}$ is the maximal subgroup $G$ such that ${\forall \ g \in G_q, \ \Phi(g,q) = q}$. 
	The identity element is
${\mathds{1} \in G}$.
\end{definition}

\begin{example}[name={The $n$-Sphere},label=example:sphere]
	For any ${n \in \mathbb{N}}$, 
	consider the sphere $\mathbb{S}^n$ and the rotation group ${SO(n\hspace{-1.5pt}+\hspace{-1.5pt}1)}$.
	For notational convenience, we model these manifolds as embedded submanifolds of Euclidean spaces, \textit{i.e.}, we let
	${	\mathbb{S}^n = \{ s \in \mathbb{R}^{n+1} : s^{\textrm{T}} s = 1\}}$ and we let
	${SO(n\hspace{-1.5pt}+\hspace{-1.5pt}1) =}$ ${ \{ R \in \mathbb{R}^{n+1 \times n+1} :  R^{\textrm{T}} R = R \, R^{\textrm{T}} = \mathds{1}
		,  \ \det R = 1 \}}$.
	Let $\rho$ and $\Psi$ be the respective restrictions to ${\mathbb{S}^n \subset \mathbb{R}^{n+1}}$ 
	of the Euclidean metric 
	and 
	of the usual action of ${SO(n\hspace{-1.5pt}+\hspace{-1.5pt}1)}$,
	\textit{i.e.},
	\begin{equation}
		\begin{gathered}
			{\Psi} : SO(n\hspace{-1.5pt}+\hspace{-1.5pt}1) \times \mathbb{S}^n \to \mathbb{S}^n, \
			(R,q) \mapsto R \, q.
			\label{sphere_action}
		\end{gathered}
	\end{equation}
	Then, it can be shown that $(\mathbb{S}^n,\Psi,\rho)$ is an ${SO(n\hspace{-1.5pt}+\hspace{-1.5pt}1)}$-homogeneous Riemannian manifold.
	Moreover,
$\rho$ induces the 
	canonical isomorphisms 
	${\rho^\flat : T\mathbb{S}^n \to T^*\mathbb{S}^n,  v_q \to  v_q{}^{\textrm{T}} \textrm{ and}}$ 
	${\rho^\sharp : T^*\mathbb{S}^n \to T\mathbb{S}^n,  f_q \to  f_q{}^{\textrm{T}}}$,
	and
			${G_q \cong SO(n)} \ \forall \ {q \in \mathbb{S}^n}$.
\end{example}

\begin{example}[name={A Lie Group},label=example:group]
	Consider any Lie group $G$,
	the left action of $G$ on itself (\textit{i.e.}, ${L : (h,g) \mapsto h g}$), and any $L$-invariant metric $\kappa_{\mathbb{I}}$. 
	Then
	by \cite[Thm. 5.38]{BulloAndLewis2004},
	there exists an inner product $\mathbb{I}$ on ${\mathfrak{g} \cong T_{\mathds{1}} G}$ 
	such that 
	${ \forall \ v_g, w_g \in TG}$,
	\begin{align}
		\kappa_{\mathbb{I}} (v_g,w_g) = \mathbb{I} \, \big(\diff L_{g^{-1}}(v_g),\diff L_{g^{-1}}(w_g)\big).
		\label{invariant_metric_on_lie_group}
	\end{align}
	Moreover,	 $(G,L,\kappa_{\mathbb{I}})$ is a homogeneous Riemannian manifold.
	$L$ is a free action, \textit{i.e.}, for ${g \neq \mathds{1}}$, the map $L_g$ has no fixed points. Thus, for any ${g \in G}$, the stabilizer is just ${G_g = \{\mathds{1}\}}$.
	It is clear that for all ${\xi \in \mathfrak{g} 
	}$ and ${\tau \in \mathfrak{g}^* 
	}$, we have
	${\kappa_{\mathbb{I}}^\flat(\xi) = \mathbb{I}^\flat(\xi)}$
	and
	${\kappa_{\mathbb{I}}^\sharp(\tau) = \mathbb{I}^\sharp(\tau)}$ (see, \textit{e.g.}, \cite[Sec. 2.3.4]{BulloAndLewis2004}).
\end{example}
\subsection{Navigation Functions}

The following functions are described in \cite{Koditschek1989}, in the more general setting of manifolds with boundary. They are useful in generating artificial potential forces for stabilization.
\begin{definition}
	On  a boundaryless manifold $Q$,
	a $0_Q$-\textit{navigation function} is a proper Morse function ${P : Q \to \mathbb{R}}$ that has the point ${0_Q \in Q}$ as its unique local minimizer.	A  navigation function is said to be \textit{perfect} if its domain admits no other navigation function with fewer critical points.
\end{definition}

For example, for any ${0_{\mathbb{S}^n} \in \mathbb{S}^n}$ and $k_P > 0$, the map
\begin{equation}
	P_{\mathbb{S}^n} : \mathbb{S}^n \to \mathbb{R}, \ q \mapsto -k_P \, ({0_{\mathbb{S}^n}})^{\textrm{T}} \, q
	\label{sphere_navigation}
\end{equation}
is a ${0_{\mathbb{S}^n}}$-navigation function. 
Following \cite{Koditschek1989},
a  perfect {$\mathds{1}\textrm{-navigation}$} function on ${SE(3)}$ can be given by 
\begingroup
\setlength\arraycolsep{3pt}
\renewcommand*{\arraystretch}{.85}
\begin{align}
	\label{se3_navigation}
	P_{SE(3)} : 
	\begin{pmatrix}
		R & x \\ 0_{1\times 3} & 1 
	\end{pmatrix} \mapsto 
\tr \big(K_R (\mathds{1} - R)\big) + 	x^\mathrm{T} K_x\, x,
\end{align}
\endgroup
where ${x \in \mathbb{R}^3}$,  ${R \in SO(3)}$, 
and $K_x, K_R$ are $3 \times 3$ symmetric positive-definite matrices where $K_R$ has distinct eigenvalues. Navigation functions on $\mathbb{R}^3$, $SE(2)$, $SO(3)$, and ${SO}(2)$ can be obtained from \eqref{se3_navigation},
and 
a navigation function on the product of boundaryless manifolds can be given by the sum of navigation functions on each factor
\cite{Cowan2007}. Thus, 
we can easily obtain an explicit navigation function on 
any product space whose factors are all $n$-spheres or closed subgroups of $SE(3)$, capturing most homogeneous spaces of interest.

\subsection{Fully-Actuated Mechanical Systems}

\begin{definition}
	\label{mechanical_system}
	A \textit{fully-actuated mechanical system} on a homogeneous Riemannian manifold ${(Q,\Phi, \kappa)}$
	has dynamics
	\begin{equation}
		\nabla_{\dot{q}} \dot{q} = \kappa^\sharp(f_q),
		\label{mechanical_system_dynamics}
	\end{equation} 
	where the state $\dot{q} \in TQ$ encompasses both the configuration and velocity, the force $f_q \in T_q^*Q$ is the control input, 
	and 
	${\nabla}$
	is the \textit{Riemannian} (or \textit{Levi-Civita}) \textit{connection} induced by $\kappa$. 
\end{definition}

\begin{remark}
	[Generality]
	Even if the system were subject to additional state-dependent forces ${F : TQ \to T^*Q}$, or it were governed by the Riemannian connection $\widetilde{\nabla}$ of a different metric $\widetilde{\kappa}$ that fails to be $\Phi$-invariant, the system can be rendered 
	in the form \eqref{mechanical_system_dynamics}
	by the static state feedback 
	\begin{align}
		f_q(\dot{q})
		&= \widetilde{\kappa}^\flat \big( \big({\widetilde{\nabla} - {\nabla}}\big)({\dot{q}},{\dot{q}}) 
		- F(\dot{q})
		+ \tilde{\kappa}^\sharp({f}_q^{\hspace{.5pt}\prime}) \big),
		\label{feedback_transformation}
	\end{align}
	where 
	$\big({\widetilde{\nabla} - {\nabla}}\big)$ is the ``difference tensor'' (see
	\cite[Prop. 4.13]{RiemannianLee}) between the Riemannian connection of $\tilde{\kappa}$ and of an invariant metric 
	${\kappa}$, 
	while ${f}_q^{\hspace{.5pt}\prime}$ is a ``virtual'' input.
	Since any homogeneous space with compact stabilizer admits an invariant metric 
	\cite[Cor. 3.18]{RiemannianLee}, 
	a tracking controller for a $\Phi$-invariant unforced system will, in practice, yield a tracking controller for the more general system via \eqref{feedback_transformation}. Also, the invariant case is common, and it simplifies the computations.
\end{remark}

\begin{definition}
	A control policy ${(t,\dot{q}) \mapsto f_q(t,\dot{q})}$ for the system \eqref{mechanical_system_dynamics} achieves \textit{almost global asymptotic tracking} of a reference trajectory ${q_d : \mathbb{R} \to Q}$ if, for each ${t_0 \in \mathbb{R}}$, there exists a residual set ${S_{t_0} \subseteq TQ}$ of full measure such that if
	${\dot{q}(t_0) \in S_{t_0}}$, then ${\dist_{TQ}\big(\dot{q}(t),\dot{q}_d(t)\big) \to 0}$ as ${t \to \infty}$. Note that $\dist_{TQ}$ is the Riemannian distance in $TQ$ corresponding to the Sasaki metric (the natural Riemannian metric on $TQ$ induced by $\kappa$). %
\end{definition}

\section{An Intrinsic, State-Valued Tracking Error}

In this section, we 
take an approach similar to that of \cite{Hampsey2022} (which dealt with \textit{kinematic} systems) to describe an intrinsic state-valued tracking error suitable for \textit{mechanical} systems.

\subsection{Tracking Error in Homogeneous Spaces}

\begin{definition}
	For an \textit{origin} ${0_Q} $ in $(Q,\Phi,\kappa)$ and a smooth curve ${q_d : \mathbb{R} \to Q}$, a $0_Q$-\textit{lift of $q_d$} is any smooth curve ${g_d : \mathbb{R} \to G}$ such that 
$		\Phi\big(g_d(t), 0_Q\big) = q_d(t) \textrm{ for all } t \in \mathbb{R}.$
\end{definition}
\begin{definition}
For any actual and reference configuration trajectories ${q, q_d : \mathbb{R} \to Q}$ and 
a given $0_Q$-lift $g_d$ of $q_d$, 
the \textit{configuration error trajectory} is the smooth curve given by	
\begin{equation}
	\label{error_curve}
	e : \mathbb{R} \to Q, \ t \mapsto 
			\Phi 
	\big(
	{g_d}(t)^{-1},
	q(t)
	\big),
\end{equation}
while the \textit{state error trajectory} is its derivative, ${\dot{e} : \mathbb{R} \to TQ}$.
\end{definition}

In \cite{Hampsey2022}, the authors use a tracking error of the form \eqref{error_curve} for kinematic (\textit{i.e.}, first-order) systems to synthesize optimal tracking controllers with local convergence. (In fact, such an error state has its roots in observer design \cite{Mahony2020}.)
We lift this definition to the tangent bundle to obtain a state-valued tracking error for mechanical (\textit{i.e.}, second-order) systems.
The following extends {\cite[Prop 4.1]{Hampsey2022}} to the second-order setting.
	\begin{proposition}
	\label{error_function_does_its_job}
	Consider any actual and reference trajectories ${q, q_d : \mathbb{R} \to Q}$.
	Let ${g_d : \mathbb{R} \to G}$ be a $0_Q$-lift of ${q_d}$, and let ${0_{TQ} \in TQ}$ be the zero tangent vector at $0_Q$. Then, for any time ${t \in \mathbb{R}}$,  
	${\dot{q}(t) = \dot{q}_d(t)}$
	if and only if
	${\dot{e}(t) = 0_{TQ}}$.
\end{proposition}

\begin{proof}
The actual and reference configuration can be expressed (suppressing the $t$-dependence for brevity) as
\begin{align}
	q &= \Phi (g_d, \Phi({g_d^{-1}}, q) ) = \Phi({g_d},e), 
	\quad 
	q_d  = \Phi(g_d,0_Q),
	\label{q_and_q_d_computation}
\end{align}
since by assumption, $g_d$ is a $0_Q$-lift of $q_d$. Thus, we have
\begin{subequations}
	\begin{align}
		\dot{q} = \diff \Phi^e (\dot{g}_d) &+ \diff \Phi_{g_d}(\dot{e}), 
		\label{q_dot_computation}
		\\
		\dot{q}_d = \diff \Phi^{0_Q} (\dot{g}_d) &+ \diff \Phi_{g_d}(0_{TQ}).
		\label{q_d_dot_computation}
	\end{align}
\end{subequations}
Assuming for sufficiency that ${\dot{e} = 0_{TQ}}$ (and therefore ${e = 0_Q}$), \eqref{q_dot_computation}-\eqref{q_d_dot_computation} immediately implies that ${\dot{q} = \dot{q}_d}$. Assuming for necessity that ${\dot{q} = \dot{q}_d}$ (and therefore ${q = q_d}$), we recall that  for each ${g \in G}$, the maps $\Phi_g$ and $\diff \Phi_g$ are automorphisms of $Q$ and $TQ$ respectively. 
For this reason, \eqref{q_and_q_d_computation} %
implies ${e = 0_Q}$, while that conclusion and \eqref{q_dot_computation}-\eqref{q_d_dot_computation}
imply that ${\dot{e} = 0_{TQ}}$.
\end{proof}

	In summary, the intrinsic tracking error smoothly transforms the state space (in a manner depending smoothly on time) in such a way that only the current reference state is mapped to the zero tangent vector over the origin of the lift.
	Also, for the ($\Phi$-invariant) Riemannian distance ${\dist_Q}$, \eqref{q_and_q_d_computation} implies that
	${\dist_Q\big(q(t),q_d(t)\big) = \dist_Q\big(e(t),0_Q\big)}$. 
	Thus, \eqref{error_curve} encodes the distance from the reference configuration ``accurately''.

	\subsection{Computing Horizontal Lifts of Curves}

The following notions will aid in computing a lift of a given reference trajectory (and ultimately, the tracking error).

\begin{definition}[see {\cite[Sec. 23.4]{GallierI}}]
	For a homogeneous Riemannian manifold $(Q,\Phi,\kappa)$ with designated origin ${0_Q \in Q}$, a 
	\textit{reductive decomposition} is a splitting ${\mathfrak{g} = \mathfrak{f} \oplus \mathfrak{q}}$, where 
	${F = \exp \mathfrak{f} = G_{0_Q}}$ (the stabilizer at $0_Q$) and
	${\mathfrak{q} \subseteq \mathfrak{g}}$ is an 
	$\Ad_F$-invariant subspace (which need not be closed under $[\hspace{.5pt}\cdot\hspace{1pt},\hspace{-1pt}\cdot\hspace{.5pt}]$).
\end{definition}

In fact, every homogeneous Riemannian manifold admits a (perhaps non-unique) reductive decomposition \cite[Prop. 1]{Kowalski2001}.

\begin{proposition}[Horizontal Lifts in Reductive Homogeneous Spaces]
	\label{horizontal_lifts_reductive}
	Consider a homogeneous Riemannian manifold $(Q,\Phi,\kappa)$ with origin 
	${0_Q \in Q}$ and 
	reductive decomposition ${\mathfrak{g} = \mathfrak{f} \oplus \mathfrak{q}}$. For each smooth curve ${q_d : \mathbb{R} \to Q}$ and initial point ${g_0 \in G}$ such that ${\Phi(g_0,0_Q) = q_d(0)}$, 
	there exists a unique smooth curve 
${g_d : \mathbb{R} \to G}$ (called the {\normalfont horizontal lift of $q_d$ through $g_0$}) 
	such that
(i)
	${g_d}$ is a $0_Q$-lift  of $q_d$,
(ii) $g_d(0) = g_0$, and
(iii) $\diff L_{g_d(t)}^{-1}\big(\dot{g}_d(t)\big) \in \mathfrak{q} $ for all $ t \in \mathbb{R}$.
\end{proposition}

\begin{proof}
	By {\cite[Prop 9.33]{GallierII}},
	 $G$ is a (left) principal $F$-bundle over ${Q = G / F}$. In particular, the projection map is given by  
	${\pi: G \to Q, \ g \mapsto \Phi(g, 0_Q)}$,
	while the free and proper action is given by
	${\Upsilon : F \times G \to G, (f,g) \mapsto R_{f^{-1}}(g)}$. Using the $\Ad_F$-invariance of $\mathfrak{q}$, it can be shown that 
	${HG = }$ ${\{\diff L_g (\mathfrak{q}) : g \in G\}}$ 
	is a principal connection on ${\pi: G \to Q}$.
	Then, the claim follows from the existence and uniqueness of horizontal lifts in principal bundles (see \cite[Sec. 2.9]{Bloch2003}).
\end{proof}

		\begin{remark}[Sections, Lifts, and Nontrivial Bundles]
			The proof of Proposition \ref{horizontal_lifts_reductive} describes a sense in which a lift projects ``down'' to the original curve via $\pi$.
			If the principal bundle ${\pi : G \to Q}$ is \textit{trivial} (\textit{i.e.}, ${G \cong Q \times  F}$ \textit{globally} and not merely \textit{locally}), then there exist \textit{global sections} ${\sigma : Q \to G}$, \textit{i.e.},  smooth maps satisfying ${\pi \circ \sigma = \id}$.
			In fact, any such section furnishes a (perhaps non-horizontal) lift ${g_d : t \mapsto \sigma \circ q_d(t)}$. However, 
	when ${\pi : G \to Q}$ is a \textit{nontrivial} bundle (\textit{e.g.},  the bundle corresponding to $(\mathbb{S}^2, \Psi, \rho)$, namely ${\pi : SO(3) \to \mathbb{S}^2}$), 
	the nonexistence of global sections
 makes it impossible to
 generate global lifts using a section. Nor can the initial value $g_{0}$ of a horizontal lift depend continuously on $q_d(0)$ alone. Thus, continuous deformation of  the reference trajectory can result in discontinuous changes in the tracking error. However, such a discontinuity is with respect to the \textit{choice} of reference trajectory (\textit{i.e.}, the planning layer); once a reference trajectory $q_d$ has been selected, the configuration error will depend smoothly on both time and state. 
	Note also that even for a \textit{horizontal} lift $g_d$ of some $q_d$, 
when
	${q_d(t_1) = q_d(t_2)}$ it may still be that  ${g_d(t_1) \neq g_d(t_2)}$ 
	due to nontrivial ``holonomy'' (see \cite[Fig. 3.14.2]{Bloch2003}),
	which cannot always be avoided.

\end{remark}

\begin{remark}[Computing Horizontal Lifts Numerically]
	\label{computing_horizontal_lifts}
Let ${(\pi \times \varphi^i):}$ ${ \pi^{-1}(U_i) \subseteq G}$ ${\to}$ ${(U_i \subseteq Q) \times F,}$ ${i \in \{1,\ldots,k\} }$ be a collection of local trivializations covering ${\pi : G \to Q}$. Let ${\mathbb{A}^i : TU_i \to \mathfrak{f}}$ be the local connection forms for $HG$ (see \cite[Prop. 2.9.12]{Bloch2003}). Suppose ${q_d(t) \in U_j}$ for ${t \in [t_1,t_2]}$ and let ${f^j : [t_1,t_2] \to F}$ solve the initial value problem (IVP)
	\begin{align}
		f^j(t_1) = \varphi^{\hspace{.75pt}j} \big(g(t_1)\big), \quad \dot{f}^j(t) = \diff L_{f^j(t)} \circ \mathbb{A}^j\big(\dot{q}_d(t)\big).
		\label{fiber_IVP}
	\end{align}
	Then the horizontal lift $g_d$ through $g_0$ of the curve $q_d$ satisfies
	\begin{align}
		g_d(t) = (\pi \times \varphi^{\hspace{.75pt}j})^{-1}\big(q_d(t),f^j(t)\big) \textrm{ for all } t \in [t_1,t_2].
		\label{lift_reconstruction}
	\end{align}
	Moreover, the restriction of the reference trajectory to any finite interval may be subdivided
	into finitely many segments, each contained within a single trivialization. 
	We repeatedly solve the IVP \eqref{fiber_IVP} numerically for each segment,
	ultimately reconstructing a smooth lift in $G$ via \eqref{lift_reconstruction}. Computationally, this is preferable to solving an IVP in $G$ directly, since it ensures that integration error will accumulate {only along the fibers} (\textit{vs.} horizontally). 
	Hence, numerical integration accuracy is not paramount (reducing the computational burden), since the solution will be an \textit{exact} lift, even if it is not perfectly {horizontal}. 
	The value of the lifted reference at each time $t$ can also be computed ``just in time'' to compute the tracking error.
\end{remark}

\begin{example}[name={The $n$-Sphere},continues=example:sphere]
		\label{kinematic_n_sphere}
\begingroup
\setlength\arraycolsep{3pt}
\setlength\abovedisplayskip{4pt}
\setlength\belowdisplayskip{4pt}
\renewcommand*{\arraystretch}{.9}
On the homogeneous Riemannian manifold $(\mathbb{S}^n,\Psi,\rho)$, we choose the origin ${0_{\mathbb{S}^n} = (0,\ldots,1) := e_n}$ and  
identify $\mathfrak{so}(n)$ with the ${n \times n}$ skew-symmetric matrices.
Following \cite[Sec. 23.5]{GallierI},
we have the 
 reductive decomposition ${\mathfrak{so}(n\hspace{-1.5pt}+\hspace{-1.5pt}1) = \mathfrak{f} \oplus \mathfrak{q}}$, where
\begin{equation*}
	\mathfrak{f} = \left\{
	\begin{pmatrix}
		\xi & 0 \\ 0 & 0 
	\end{pmatrix} : \xi \in \mathfrak{so}(n)
	\right\}, \  
	\mathfrak{q} = 
	\left\{
\begin{pmatrix}
	0 & \upsilon \\ -\upsilon^{\textrm{T}} & 0 
\end{pmatrix} : 
\upsilon \in \mathbb{R}^{{n}}	
\right\}.
\end{equation*}
We horizontally lift a given configuration reference trajectory ${q_d : \mathbb{R} \to \mathbb{S}^n}$
in the sense of Proposition \ref{horizontal_lifts_reductive} to obtain the lifted reference ${R_d : \mathbb{R} \to {SO(n\hspace{-1.5pt}+\hspace{-1.5pt}1)}}$, performing numerical computations in the manner of Remark \ref{computing_horizontal_lifts} for accuracy. \endgroup
Then, 
the configuration error in the form \eqref{error_curve} is simply given by
\begin{align}
	e : \mathbb{R} 
	\to \mathbb{S}^n, \ t \mapsto {R_d}(t)^{\textrm{T}} \, q(t).
	\label{sphere_error_map}
\end{align}
	\end{example}
\begin{example}[name={A Lie Group},continues=example:group]	
	Since the stabilizer of any point $g$ on the  homogeneous Riemannian manifold $(G,L,\kappa_{\mathbb{I}})$ is ${G_g = \{\mathds{1}\}}$, $0_G$-lifts (of any kind) are unique for each ${0_G \in G}$. Making
	the usual
	choice
	${0_G = \mathds{1}}$, the lift is the original curve itself,
	and the configuration error is
	\begin{equation}
		e : \mathbb{R} \to G, \ t \mapsto {g_d}(t)^{-1} \, g(t),
		\label{lie_group_error_map}
	\end{equation}
	showing that in the special case of Lie groups, the configuration error reduces elegantly to a familiar, intuitive form, called the ``right group error function'' \cite[p. 548]{BulloAndLewis2004}. In the additive group $G = (\mathbb{R}^n, +)$ (\textit{i.e.}, a vector space), the configuration error is (unsurprisingly) the map ${t \mapsto  g(t) - g_d(t)}$.
\end{example}

\section{Almost Global Asymptotic Tracking}

We now present the main result on trajectory tracking. First, we define the \textit{body velocity} of any curve ${\gamma : \mathbb{R} \to G}$ to be the curve in $\mathfrak{g}$ 
given by ${t \mapsto \diff L_{\gamma(t)}^{-1} \big( \hspace{.5pt} \dot{\gamma}(t) \hspace{.5pt} \big)}$.

\begin{theorem}
	Consider a fully-actuated mechanical system on $(Q, \Phi, \kappa)$, \label{mechanical_tracking_control_theorem} a reference trajectory ${q_d : \mathbb{R} \to Q}$, and any ${0_Q}$-lift ${g_d : \mathbb{R} \to G}$ of $q_d$ with bounded body velocity.
	Let ${P}$ be a ${0_Q}$-navigation function and $\nu$ be a Riemannian metric on $Q$. For each 
	(fixed) 
	state 
	${\dot{q} \in TQ}$, 
	let 
	$
		{		e^q : 
			r \mapsto \Phi_{g_d(r)}^{-1} \big(q\big)
		}%
	$
	be a curve
	 in $Q$ and
	$
	{%
			X^{\dot{q}} : }$ ${r \mapsto \diff \Phi_{g_d(r)}^{-1} (\dot{q})
	}$
	be a vector field along $e^q$.
	Then, the control  policy
\begin{equation}
	\begin{aligned}
		\label{mechanical_tracking_control}
		\hspace{-1pt} 
		f_q(t,\dot{q})
		=
		-\diff  \Phi^*_{g_d^{-1}}
		\big(
		\diff &P(e)
+
		\nu^\flat(\dot{e})
\, +
		\\ &
		\kappa^\flat 
		\circ 
					\big(
		\nabla_{\dot{e}^{q}} 
		( \dot{e}^{q}+ 2 \, 
		X^{\dot{q}}
		)
					\big)
		(t)
		\big)
	\end{aligned}
\end{equation}
		achieves almost global asymptotic tracking of the reference and local exponential convergence of the tracking error.
	\end{theorem}

\begin{remark}
	We use the ``dummy variable'' $r$ 
	in  $e \raisebox{3.5pt}{${\scriptstyle q}$}$ and $X \raisebox{3.5pt}{${\scriptstyle \dot{q}}$}$ 
	to emphasize that $q$ and $\dot{q}$ are held fixed as $r$ varies (despite their dependence on $t$). This allows us to rigorously and intrinsically express \eqref{mechanical_tracking_control} using only the standard formalism for covariant differentiation along curves. It follows from \cite[Prop. 4.26]{RiemannianLee} that \eqref{mechanical_tracking_control} depends only on ${t \in \mathbb{R}}$ and ${\dot{q} \in TQ}$.
		Additionally, although the control policy \eqref{mechanical_tracking_control} requires \textit{choosing} a certain lift $g_d$, we will show that the qualitative closed-loop stability properties are ultimately independent of this choice.
	\end{remark}
	
	\begin{proof}[Proof of Theorem \ref{mechanical_tracking_control_theorem}]
		The configuration error \eqref{error_curve} is
		the ``diagonal'' curve of the smooth ``family of curves'' given by
		\begin{equation}
			E : \mathbb{R}^2 \to Q, \,  (r,s) \mapsto \Phi\big(g_d(r)^{-1}, q(s)\big),
		\end{equation}
		in the sense that ${e(t) = E(t,t)}$. We will use this observation to express the covariant derivative $\nabla_{\dot{e}} \dot{e}$ in terms of the respective contributions of the reference and actual trajectories. 
		
		Following \cite[Ch. 6]{RiemannianLee}, the family of curves $E$ has \textit{transverse} and \textit{main} curves at each $s$ and $r$, 
		given respectively by
${E^s : r \mapsto E(r,s)}$ and ${E_r : s \mapsto E(r,s).}$
		Let $\mathfrak{X}(E)$ be the set of vector fields over $E$, \textit{i.e.}, maps ${(r,s) \mapsto V(r,s)}$ such that ${V(r,s) \in T_{E(r,s)} Q}$. For example, two such vector fields are the \textit{transverse} and \textit{main velocity},
		${\partial_r E : (r,s) \mapsto \dot{E}^s (r)}$ and 
		${\partial_s E : (r,s) \mapsto \dot{E}_r (s)}$.
		For any ${V \in \mathfrak{X}(E)}$, we may compute its ``partial'' covariant derivative along the transverse or main direction, operations denoted respectively by ${\Diff_r, \, \Diff_s : \mathfrak{X}(E) \to \mathfrak{X}(E)}$. This operation is defined by restricting the vector field over the family of curves to a vector field along each transverse (resp. main) curve and computing the usual covariant derivative along that curve. 
		
		Covariant differentiation is a purely local operation \cite[Lemma 4.1]{RiemannianLee}, and thus a local coordinate calculation can be used to show that because $e$ is the diagonal curve of $E$, 
		\begin{align}
			\big(
			\nabla_{\dot{e}} \dot{e}
			\big)(t)
			&=
			\big(
			\Diff_r \partial_r \mathit{E}
			+ 2 \Diff_r \partial_s \mathit{E}
			+ \Diff_s \partial_s \mathit{E}
			\big)(t,t).
			\label{application_of_diagonal_lemma}
		\end{align}
		where we have evaluated the right-hand side at $(r,s) = (t,t)$.
		We now aim to compute the terms on the right-hand side of \eqref{application_of_diagonal_lemma}.
		Observing that 
		${E^s (r) = e^{q(s)}}(r)$
		and also that
		\begin{align}
	\partial_s E(r,s) = 
	\diff\Phi_{g_d(r)^{-1}} \big(\dot{q}(s)\big) = X^{\dot{q}(s)}(r),
	\label{partial_s_E}
\end{align}
		we may verify that
		\begin{align}
			\big(\Diff_r \partial_r E\big) (t,t) &= \big(\nabla_{\dot{E}^t} {\dot{E}^t}\big)(t)
			= \big(\nabla_{\dot{e}^q} {\dot{e}^q}\big)(t),
			\label{D_r_p_r_E}
			\\
			\big(
			\Diff_r 
			\partial_s \mathit{E}
			\big)(t,t) 
			&=
			(\nabla_{\dot{e}^{q}} 
			X^{\dot{q}}
			)(t).
			\label{D_r_p_s_E}
		\end{align}	
		By \cite[Thm. 5.70]{BulloAndLewis2004}, 
		the $\Phi$-invariance of $\kappa$ implies that 
		$\nabla$ satisfies 
		$
			{\nabla_{\diff \Phi_g(X)} {\diff \Phi_g(Y)} =  \diff \Phi_g \big(\nabla_XY\big)}$ 
		for all ${X,Y \in \mathfrak{X}(Q)}$ and ${g \in G}$. %
		From this fact, \eqref{mechanical_system_dynamics}, and \eqref{partial_s_E}, it follows that
		\begin{align}
			(\Diff_s \partial_s E) (t,t) = \big(\nabla_{\dot{E}_t} {\dot{E}_t}\big)(t)
			&=
			\diff \Phi_{g_d(t)}^{-1}  
			\circ \kappa^\sharp\big(f_q(t)\big),
			\label{D_s_p_s_E}
		\end{align}
		where ${t \mapsto f_q(t)}$ is the input force signal corresponding to ${t \mapsto q(t)}$.
		Substituting these results into \eqref{application_of_diagonal_lemma}, we obtain
		\begin{equation}
			\begin{aligned}
				\hspace{-4pt}(
				\nabla_{\dot{e}} \dot{e}
				)(t)
				=
				\nabla_{\dot{e}^{q}}
				( \dot{e}^{q} + 2 \,	X^{\dot{q}}
				)
				(t)
				+ \diff \Phi_{g_d(t)}^{-1}  
				\circ \kappa^\sharp\big(f_q(t)\big).
				\label{covariant_acceleration_of_error}
			\end{aligned}
		\end{equation}
		
		Substituting \eqref{mechanical_tracking_control} into \eqref{covariant_acceleration_of_error} and using the $\Phi$-invariance of $\kappa$ to verify that 
		${(\diff \Phi_{g} \circ \kappa^\sharp)^{-1} = \diff \Phi_g^* \circ \kappa^\flat}$  for any ${g \in G}$ ultimately  yields the autonomous tracking error dynamics
		\begin{align}
			\nabla_{\dot{e}} \dot{e} = 
- \kappa^\sharp \big(\diff P (e) + \nu^\flat(\dot{e}) \big) ,
			\label{mechanical_error_dynamics}
		\end{align}
		which is a mechanical system with strict dissipation and a navigation function potential. Thus,
		it follows from \cite[Thm. 2]{Koditschek1989} that 
		there exists an open dense set ${S \subseteq TQ}$ of full measure, such that for any ${t_0 \in \mathbb{R}}$ (since \eqref{mechanical_error_dynamics} is autonomous), if 
${\dot{e}(t_0) \in S}$, then ${\dot{e}(t) \to 0_{TQ}}$ as ${t \to \infty}$, and local exponential stability $0_{TQ}$ for \eqref{mechanical_error_dynamics} follows from \cite[Thm. 6.45]{BulloAndLewis2004}.

Finally, it follows from \cite[II.A.2]{EpsteinMarden2006} that ${\dist_{TQ}(
u_p,
v_r
) 
\leq }$ ${
\dist_Q(p,r)
+ 
\big|\big|\tau_{\gamma}(u_p) - {v}_r\big|\big|_\kappa}$ for sufficiently close points ${p,r \in Q}$, where $\tau_\gamma$ denotes parallel transport along the shortest geodesic from $p$ to $r$. The $\Phi$-invariance of $\kappa$, the boundedness of the reference body velocity, and \eqref{q_dot_computation}-\eqref{q_d_dot_computation} can then be used 
with this inequality
to verify that ${\dot{e}(t) \to 0_{TQ}}$ implies ${\dist_{TQ}\big(\dot{q}(t),\dot{q}_d(t)\big) \to 0}$%
.  Since $\diff \Phi_{g}$ is diffeomorphism, \eqref{q_dot_computation} implies (for each ${t_0 \in \mathbb{R}}$) the existence an open dense set ${S_{t_0} \subseteq TQ}$ of full measure such that if ${\dot{q}(t_0) \in S_{t_0}}$, then 
${\dist_{TQ}\big(\dot{q}(t),\dot{q}_d(t)\big) \to 0}$ as ${t \to \infty}$, proving the claim.
	\end{proof}
			
		\section{Specialization to $n$-Spheres and Lie Groups}
	
We now specialize the controller proposed in Theorem \ref{mechanical_tracking_control_theorem} to the $n$-sphere and Lie group settings, in each case obtaining a concise and explicit expression for the control policy \eqref{mechanical_tracking_control}.

		\begin{corollary}[Almost Global Asymptotic Tracking on $\mathbb{S}^n$]
	For any ${n \in \mathbb{N}}$, consider a fully-actuated mechanical system on  ${(\mathbb{S}^n,\Psi,\rho})$, 
		a reference trajectory ${q_d : \mathbb{R} \to \mathbb{S}^n}$, 
	and 
	any ${0_{\mathbb{S}^n}}$-lift ${R_d : \mathbb{R} \to SO(n\hspace{-1.5pt}+\hspace{-1.5pt}1)}$ of 
$q_d$ with bounded body velocity.
	Then for any ${k_P, k_D > 0}$,
	the control policy
	\begin{equation}
		\begin{aligned}
			\hspace{-6pt} f_q(t,\dot{q})   =
			-k_P &  \, {q_d}^{\textrm{\normalfont T}} \big(q \,  q^\textrm{\normalfont T} - \mathds{1} \big)
			- k_D \big(\dot{q}^{\textrm{\normalfont T}} + q^{\textrm{\normalfont T}} {\dot{R}_d} \, {R_d}^{\textrm{\normalfont T}}\big)
			\\&  
			+ \big(
			q^{\textrm{\normalfont T}} {\ddot{R}_d} + 2 \, {\dot{q}}^{\textrm{\normalfont T}} \dot{R}_d \big)
			\,
			{R_d}^{\textrm{\normalfont T}} 
			\big(q \,  q^\textrm{\normalfont T} - \mathds{1} \big)
			\label{sphere_mechanical_tracking_controller}
		\end{aligned}
	\end{equation}
	achieves almost global asymptotic tracking of the reference and local exponential convergence of the tracking error.
	\label{sphere_tracking_corollary}
\end{corollary}

\begin{proof}
	It will suffice to apply Theorem \ref{mechanical_tracking_control_theorem} with the configuration error \eqref{sphere_error_map}, the navigation function \eqref{sphere_navigation}, and the dissipation metric $\nu = k_D \, \rho$.
	To show this, we compute
	\begin{align}
		\diff P(e) &=(-k_P {0_{\mathbb{S}^n}})^{\textrm{T}} (\mathds{1} - e\, e^\textrm{\normalfont T} \hspace{1pt} ),
		\\
		\nu^\flat(\dot{e}) &= 
		k_D (
		q^\textrm{T} \, \dot{R}_d + \dot{q}^{\textrm{T}} \, R_d 
		).		
	\end{align}
	Moreover, the covariant derivative along $\gamma$ of $X \in \mathfrak{X}(\gamma)$ is
	\begin{equation}
		\nabla_{\dot{\gamma}} X : r \mapsto \big(\mathds{1} - \gamma(r) \, \gamma(r)^\textrm{T}\big) \, \dot{X}(r),
	\end{equation}
	and we have ${e^q(r) = R_d(r)^{\textrm{T}} \, q}$ and ${X^{\dot{q}}(r) = R_d(r)^{\textrm{T}} \, \dot{q}}$. Thus, 
	\begin{align}
		\nabla_{\dot{e}^q} (
		\dot{e}^q + 2 \, X^{\dot{q}}			
		) (t) 
		=
		\big(\mathds{1} - e\, e^\textrm{\normalfont T} \hspace{1pt} \big)
		\big(
		\ddot{R}_d{}^{\textrm{T}} \, q + 2 \, \dot{R}_d{}^{\textrm{T}} \, \dot{q}			
		\big).
	\end{align}
	Finally, noting that ${\diff \Psi^*_R(\omega) = \rho^\flat \circ \diff \Psi_{R}^{-1} \circ \rho^\sharp(\omega)  = \omega \, {R}^{\textrm{T} }}$ due to the $\Psi$-invariance of $\rho$, and also verifying that ${\big(\mathds{1} - e\, e^\textrm{\normalfont T} \hspace{1pt} \big)
		=
		{R_d}^{\textrm{T}}
		\big(\mathds{1} -  q \, q^{\textrm{T}} \big) {R_d}}, 
	$ we may compute \eqref{mechanical_tracking_control} using all the previous calculations, simplifying to yield exactly \eqref{sphere_mechanical_tracking_controller}.
	Thus, the claim follows immediately by Theorem \ref{mechanical_tracking_control_theorem}.
\end{proof}
For any function ${P : G \to \mathbb{R}}$
on a Lie group $G$,  
we define a map ${\zeta_P : G \to \mathfrak{g}^*}, \ {g \mapsto \diff L_g^* \circ \diff P (g)}$. 
Also, we denote the body velocities of trajectories $g$, $g_d$, and $e$ by by $\xi$, $\xi_d$, and $\xi_e$ respectively.
The control policy we now recover uses the same configuration error as in \cite[Thm. 11.29]{BulloAndLewis2004}, although our different feedforward terms lead to almost global tracking. In that sense, our result is qualitatively more similar to \cite{Maithripala2006}, but they use a different configuration error (\textit{i.e.}, $e = g_d \, g^{-1}$).

\begin{corollary}[Almost Global Asymptotic Tracking on a Lie Group]
	For any Lie group $G$, 
	consider a fully-actuated mechanical system on ${(G,L,\kappa_{\mathbb{I}})}$
	and a reference trajectory ${g_d : \mathbb{R} \to G}$ with bounded body velocity.
	Let ${P}$ be a ${\mathds{1}}$-navigation function on $G$
	and
	$\mathbb{D}$ be an inner product on $\mathfrak{g}$. Then for the virtual control ${\tau = \diff L_{g}^* (f_g)}$,
the control policy
		\begin{equation}
			\begin{aligned}
				\hspace{-7pt}
				\tau(t,g,\xi)
				\hspace{-1pt}
				&= 
				\hspace{-1pt}
				- \zeta_P (e) 
				- \mathbb{D}^\flat \xi_e \ +
				\\  & \quad  \ \,
				\mathbb{I}^\flat (
				\hspace{-1pt}
				\Ad_{e}^{-1} \hspace{-1pt} \dot{\xi}_d + 
				[\xi,\xi_e]	
				) + 
				\ad^*_{\xi_e} \hspace{-1pt} \mathbb{I}^\flat \xi_e - 
				\ad^*_{\xi} \hspace{-1pt} \mathbb{I}^\flat \xi 
				\hspace{-2pt}
				\label{lie_group_tracking_control}
			\end{aligned}
		\end{equation}	
		achieves almost global asymptotic tracking of the reference and local exponential convergence of the tracking error.
		\label{group_tracking_corollary}
	\end{corollary}
	
	\begin{proof}
		It will suffice to apply Theorem \ref{mechanical_tracking_control_theorem} with the configuration error \eqref{lie_group_error_map}, the given navigation function, and the metric $\nu$ induced by $\mathbb{D}$ as in \eqref{invariant_metric_on_lie_group}.
		Since $e^g(r) = L_{g_d(r)}^{-1}(g)$,
		we have
		\begin{align}
			\dot{e}^g(r) 
			&= \diff L_{e^g(r)} \big( -\Ad_{e^g(r)}^{-1}  \big(\xi_d(r)\big)  \big)
			, \\ 
			X^{\dot{g}}(r) = \diff L_{g_d(r)}^{-1}(\dot{g})
			&= \diff L_{e^g(r)} \big(\xi\big),
		\end{align}
		so that altogether, we have
		\begin{align}
			(\dot{e}^g + 2 X^{\dot{g}})(r)
			&= 
			\diff L_{e^g(r)} \big(
			2 \, \xi - \Ad_{e^g(r)}^{-1} \big(\xi_d(r)\big)
			\big).
		\end{align}
		The result \cite[Thm. 5.40]{BulloAndLewis2004} implies that
		that for a curve $\gamma$ in $G$, a vector field ${X \in \mathfrak{X}(\gamma)}$, and curves $\upsilon,\eta$ in $\mathfrak{g}$ for which ${\dot{\gamma}(r) = \diff L_{\gamma(r)}\big(\upsilon (r)\big)}$ and ${X(r) = \diff L_{\gamma(r)} \big( \eta(r) \big)}$, we have
		\begin{align}
			\big(\nabla_{\dot{\gamma}} X\big)(r) 
			= 
			\diff L_{\gamma(r)} \big(
			\dot{\eta}(r) + \nabla_{\mathbb{I}} \big({\upsilon(r)}, \eta(r)\big)
			\big),
			\label{lie_group_covariant_derivative}
		\end{align}
		where ${\nabla_{\mathbb{I}} : \mathfrak{g} \times \mathfrak{g} \to \mathfrak{g}}$ is the bilinear map given by 
		\begin{equation}
			(\upsilon,\eta) \mapsto 
			\tfrac{1}{2} [\upsilon,\eta] 
			- \tfrac{1}{2}\mathbb{I}^{\sharp}\big(
			\ad_{\upsilon}^* \mathbb{I}^\flat \, \eta + \ad_{\eta}^* \mathbb{I}^\flat \, \upsilon
			\big).
			\label{nabla_I_bilinear_map}
		\end{equation}
		Thus, after verifying that $\xi_e = \xi - \Ad_{e}^{-1} \xi_d$
		and computing
		\begin{align}
			\tfrac{\diff}{\diff r} \big(2 \, \xi &- \Ad_{e^g(r)}^{-1}\xi_d(r)\big)
			=
			-\Ad_{e^g(r)}^{-1} \dot{\xi}_d(r),
		\end{align}
		we 
		use \eqref{lie_group_covariant_derivative} and the bilinearity of \eqref{nabla_I_bilinear_map}
			to conclude that
		\begin{equation}
			\begin{aligned}
				\label{evaluate_lie_group_covariant_term}
				\hspace{-4pt}
				&
				\nabla_{\dot{e}^g}  (\dot{e}^g + 2 X^{\dot{g}})(t)
				\hspace{-2pt}
				\\
				\hspace{-2pt}
				&=\hspace{-2pt}
				-\diff L_{e} \big( \hspace{-2pt}
				\Ad_{e}^{-1} \hspace{-1.5pt} \dot{\xi}_d \hspace{-1pt} + \hspace{-1pt} 
				[\xi,\xi_e] + \mathbb{I}^\sharp(\ad_{\xi_e}^* \mathbb{I}^\flat \xi_e - \ad_{\xi}^* \mathbb{I}^\flat \xi)
				\big).
				\hspace{-4pt}
			\end{aligned}
		\end{equation}%
		Finally, we define $f_g$ as in \eqref{mechanical_tracking_control} and use \eqref{evaluate_lie_group_covariant_term} 
		to compute 
		${\tau = \diff L_{g}^*(f_g)}$, simplifying using the identities
		${\diff L_{g}^* \circ \diff L_{g_d{}^{-1}}^* = \diff L_{e}^*}$ and ${\diff L_{e}^* \circ \nu^\flat = \nu^\flat \circ \diff L_{e}^{-1}}$ to yield \eqref{lie_group_tracking_control}. Thus, 
		the claim follows immediately by Theorem \ref{mechanical_tracking_control_theorem}.
	\end{proof}

\section{Applications of the Method}

\begin{example}[The Axisymmetric Satellite]
	\label{axisymmetric_satellite_example}
	Consider a free-floating satellite, modeled as an (underactuated) mechanical system on $(SO(3),L,\kappa_{\mathbb{J}})$ consisting of a rigid body with inertia tensor ${\mathbb{J} = \mathrm{diag}(J_{1},J_{2},J_{3})}$ with ${J_1 = J_2}$ and control torques lying in a two-dimensional left-invariant codistribution corresponding to ${\spn \{\hat{e}_1,\hat{e}_2\}\subset \mathfrak{so}(3)^*}$, 
	where ${\hat{\cdot} : \mathbb{R}^3 \to \mathfrak{so}(3)^* \cong \mathfrak{so}(3)}$ is the usual isomorphism satisfying $\hat{a} b = a \times b$. The satellite is equipped with a camera or antenna aligned with the $e_3$ axis, whose bearing is thus described by the output ${y : \dot{R} \mapsto R \, e_3}$.
	This system is invariant under the (left) action of $SO(2)$ on $SO(3)$ corresponding to a \textit{body-fixed} rotation around $e_3$. 
By Noether's theorem, the evolution around the symmetry axis is governed by conservation of momentum (\textit{i.e.} ${\frac{\diff}{\diff t}\big(J_3 \,  {\Omega}_3(t)\big) = 0}$). 
When ${\Omega_3(t_0) = 0}$, reduction by this symmetry (see \cite[Thm. 5.83]{BulloAndLewis2004}) will yield a fully-actuated mechanical system on ${(\mathbb{S}^2 = SO(3) / SO(2), \Psi, J_1 \, \rho)}$.
	It is clear that output tracking for the original (underactuated) system on $SO(3)$ amounts to asymptotically tracking a state trajectory for the reduced (fully-actuated) system on $\mathbb{S}^2$.
	We apply Corollary \ref{sphere_tracking_corollary} to synthesize the tracking controller, using the reductive decomposition ${\mathfrak{so}(3) = \mathfrak{f} \oplus \mathfrak{q}}$ with  
	 ${\mathfrak{f} = \spn\{\hat{e}_3 \} \cong \mathfrak{so}(2)}$ and ${\mathfrak{q} = \spn \{	\hat{e}_1,\hat{e}_2	\}}$ to compute 
	a horizontally lifted reference 
	in the manner of Proposition \ref{horizontal_lifts_reductive} (in fact, the principal connection is the \textit{``mechanical connection''} \cite[Sec. 3.10]{Bloch2003}), using local trivializations (see Remark \ref{computing_horizontal_lifts}) obtained from \cite[Eqn. (29)-(30)]{Welde2023}.   
	Fig. \ref{sphere_tracking_figure} shows the trajectory tracking performance.

\end{example}

\begin{example}[The Omnidirectional Aerial Robot \cite{Brescianini2016}]
	Consider an aerial robot consisting of a single rigid body and actuators capable of applying arbitrary wrenches. Compensating for gravity yields 
	a fully-actuated mechanical system on 
	$(\mathbb{R}^3 \times SO(3),L,\kappa_\mathbb{I})$,
	where 
	${\mathbb{I} = }$ ${\mathrm{diag}(m \mathds{1}_{3 \times 3}, \mathbb{J}_{3 \times 3} 
		)}$. We use Corollary \ref{group_tracking_corollary} with \eqref{se3_navigation} to track trajectories (see Fig. \ref{se3_tracking_figure}).
\end{example}

\begin{figure}[t]
	\begin{center}
		\includegraphics[width=.96\columnwidth,trim={1.5cm 8.6cm 2.2cm 6.5cm}, clip]{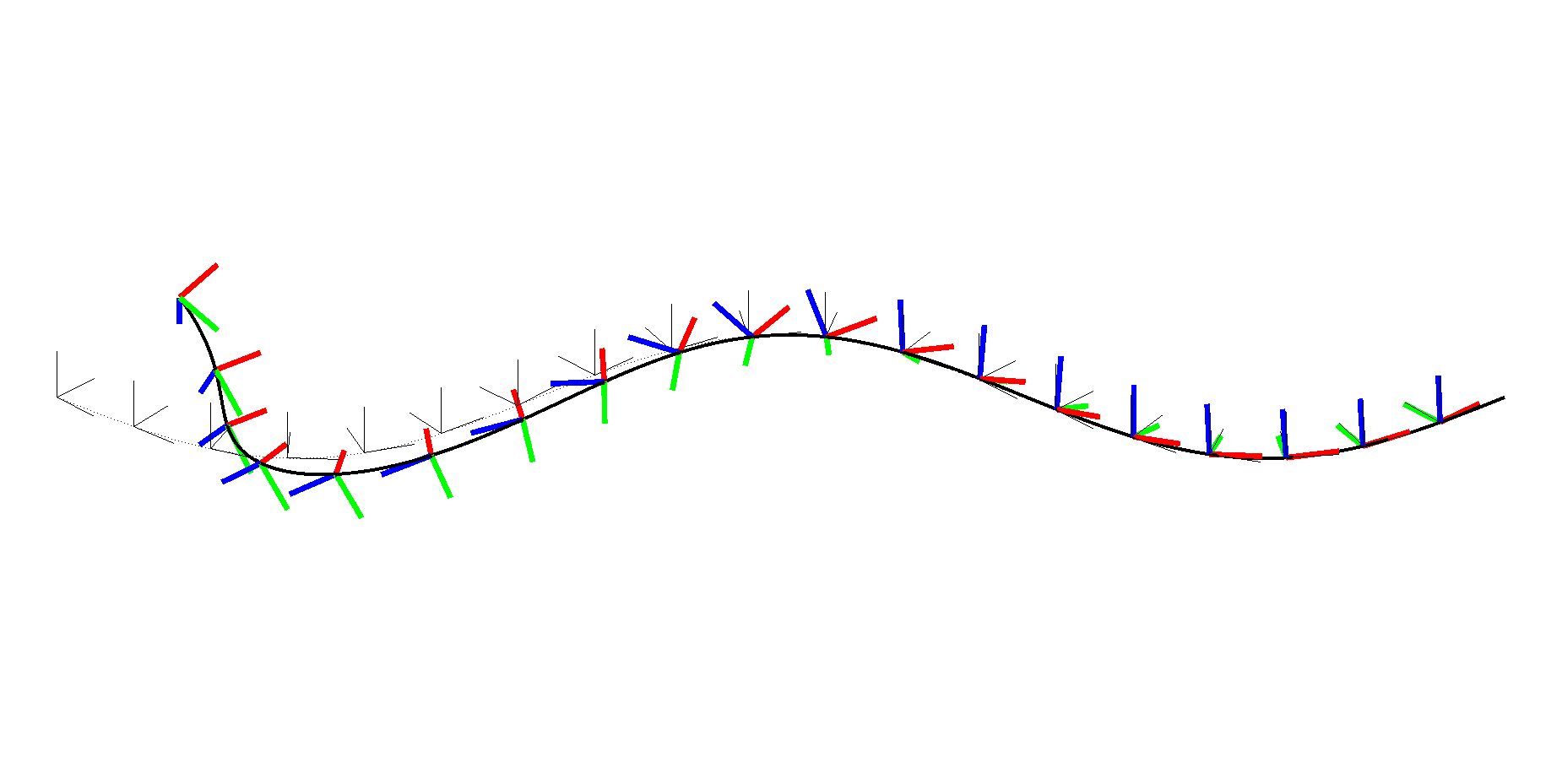}%
		\vspace{-10pt}
	\end{center}
	\caption{
		Almost global asymptotic trajectory tracking on ${\mathbb{R}^3 \times SO(3)}$.
	} 
	\vspace{-14pt}
	\label{se3_tracking_figure}
\end{figure}

\section{Discussion}
	Of the $n$-spheres, only $\mathbb{S}^1$ and $\mathbb{S}^3$ admit a Lie group structure, and thus methods like \cite{Maithripala2006} are not applicable to tracking on, \textit{e.g.}, $\mathbb{S}^2$. 
	As pointed out in \cite{Nayak2019}, methods which do not rely on reduction to regulation often fail to achieve or certify almost global convergence, in large part because they cannot benefit from \cite[Thm. 2]{Koditschek1989}.
	Other smooth tracking controllers on $\mathbb{S}^2$ (\textit{e.g.} 
	\cite{Bullo1999}
	or 
	\cite[Sec. III]{Lee2016}) do \textit{not} guarantee convergence from almost every initial state in $TQ$, rather ensuring convergence from a smaller region that, \textit{e.g.}, may contain the \textit{zero} tangent vector over almost every point in $Q$ ({not} $TQ$), sometimes leading to an abuse of the term ``almost global'' in regards to asymptotic stability.
 We also note that the differential properties of a state-valued tracking error are important, since it was the surjectivity of the partial derivative of the tracking error appearing in \eqref{covariant_acceleration_of_error} that enabled the feedforward cancellation of other terms and the injection of arbitrary dissipation and potential, regardless of the system considered (a difficulty of the approach proposed in \cite{Nayak2019}).

Example \ref{axisymmetric_satellite_example} demonstrates applicability to output tracking for certain underactuated systems. Moreover, hierarchical controllers for underactuated systems (\textit{e.g.}, \cite{Lee2010}) often rely on the control of subsystems that ``look'' fully-actuated, and the identification of a {``geometric flat output''} \cite{Welde2023} can often yield such a decomposition wherein the subsystems evolve on homogeneous spaces that may \textit{not} be Lie groups. Also, the almost global asymptotic stability of a hierarchical controller with subsystems of the form proposed in ${\textrm{Theorem \ref{mechanical_tracking_control_theorem}}}$ can be certified \textit{compositionally} using only the properties of the subsystems and the {``interconnection term''} \cite{Welde2023b}, thanks to the dissipative mechanical form of the error dynamics \eqref{mechanical_error_dynamics}.
 For these reasons, we believe the proposed method has substantial implications in the systematic synthesis of controllers 
 for a broad class of underactuated systems.

\section{Conclusion}

We propose a systematic, unified 
tracking controller for fully-actuated mechanical systems evolving on homogeneous Riemannian manifolds, with exactly zero probability of nonconvergence from a randomly selected initial state. 
We apply the method to systems with two different configuration manifolds.
Conceptually, our results illustrate that it is the \textit{transitivity} of a Lie group's action on the configuration manifold (and \textit{not} the absence of fixed points) 
that is useful for tracking control via error regulation in mechanical systems.

\bibliographystyle{IEEEtran}
\bibliography{IEEEabrv,refs}

\begin{thebibliography}{10}
\providecommand{\url}[1]{#1}
\csname url@rmstyle\endcsname
\providecommand{\newblock}{\relax}
\providecommand{\bibinfo}[2]{#2}
\providecommand\BIBentrySTDinterwordspacing{\spaceskip=0pt\relax}
\providecommand\BIBentryALTinterwordstretchfactor{4}
\providecommand\BIBentryALTinterwordspacing{\spaceskip=\fontdimen2\font plus
\BIBentryALTinterwordstretchfactor\fontdimen3\font minus
  \fontdimen4\font\relax}
\providecommand\BIBforeignlanguage[2]{{%
\expandafter\ifx\csname l@#1\endcsname\relax
\typeout{** WARNING: IEEEtran.bst: No hyphenation pattern has been}%
\typeout{** loaded for the language `#1'. Using the pattern for}%
\typeout{** the default language instead.}%
\else
\language=\csname l@#1\endcsname
\fi
#2}}

\bibitem{Nayak2019}
A.~Nayak and R.~N. Banavar, ``On almost-global tracking for a certain class of
  simple mechanical systems,'' \emph{IEEE Transactions on Automatic Control},
  vol.~64, no.~1, pp. 412--419, 2019.

\bibitem{Matni2024}
N.~Matni, A.~D. Ames, and J.~C. Doyle, ``Towards a theory of control
  architecture: A quantitative framework for layered multi-rate control,''
  \emph{arxiv:2401.15185}, 2024.

\bibitem{Koditschek1989}
D.~E. Koditschek, ``The application of total energy as a {L}yapunov function
  for mechanical control systems,'' \emph{Contemporary mathematics}, vol.~97,
  p. 131, 1989.

\bibitem{Maithripala2006}
D.~S. Maithripala, J.~M. Berg, and W.~P. Dayawansa, ``Almost-global tracking of
  simple mechanical systems on a general class of {L}ie groups,'' \emph{IEEE
  Transactions on Automatic Control}, vol.~51, no.~2, pp. 216--225, 2006.

\bibitem{Bernuau2013}
E.~Bernuau, W.~Perruquetti, and E.~Moulay, ``Retraction obstruction to
  time-varying stabilization,'' \emph{Automatica}, vol.~49, no.~6, pp.
  1941--1943, 2013.

\bibitem{Hampsey2022}
M.~Hampsey, P.~{van Goor}, and R.~Mahony, ``Tracking control on homogeneous
  spaces: the {E}quivariant {R}egulator {(EqR)},'' \emph{IFAC-PapersOnLine},
  vol.~56, no.~2, pp. 7462--7467, 2023.

\bibitem{RiemannianLee}
J.~M. Lee, \emph{Introduction to {R}iemannian manifolds}, ser. Graduate Texts
  in Mathematics.\hskip 1em plus 0.5em minus 0.4em\relax Springer, 2018, vol.
  176.

\bibitem{BulloAndLewis2004}
F.~Bullo and A.~D. Lewis, \emph{Geometric Control of Mechanical Systems}, ser.
  Texts in Applied Mathematics.\hskip 1em plus 0.5em minus 0.4em\relax Springer
  Verlag, 2004, vol.~49.

\bibitem{Cowan2007}
N.~J. Cowan, ``Navigation functions on cross product spaces,'' \emph{IEEE
  Transactions on Automatic Control}, vol.~52, no.~7, pp. 1297--1302, 2007.

\bibitem{Mahony2020}
R.~Mahony, T.~Hamel, and J.~Trumpf, ``Equivariant systems theory and observer
  design,'' \emph{arXiv:2006.08276}, 2020.

\bibitem{GallierI}
J.~Gallier and J.~Quaintance, \emph{Differential Geometry and Lie Groups: A
  Computational Perspective}, ser. Geometry and Computing.\hskip 1em plus 0.5em
  minus 0.4em\relax Springer, 2020, vol.~12.

\bibitem{Kowalski2001}
O.~Kowalski and J.~Szenthe, ``On the existence of homogeneous geodesics in
  homogeneous {R}iemannian manifolds,'' \emph{Geom. Dedicata}, vol.~81, no.
  1-3, pp. 209--214, 2000.

\bibitem{GallierII}
J.~Gallier and J.~Quaintance, \emph{Differential Geometry and Lie Groups, A
  Second Course}, ser. Geometry and Computing.\hskip 1em plus 0.5em minus
  0.4em\relax Springer, 2020, vol.~12.

\bibitem{Bloch2003}
A.~M. Bloch, \emph{Nonholonomic mechanics and control}, ser. Interdisciplinary
  Applied Mathematics.\hskip 1em plus 0.5em minus 0.4em\relax Springer-Verlag,
  2003, vol.~24.

\bibitem{EpsteinMarden2006}
D.~B.~A. Epstein and A.~Marden, ``Convex hulls in hyperbolic space, a theorem
  of {S}ullivan, and measured pleated surfaces,'' in \emph{Fundamentals of
  Hyperbolic Manifolds: Selected Expositions}, ser. London Math. Soc. Lecture
  Note Ser.\hskip 1em plus 0.5em minus 0.4em\relax Cambridge Univ. Press, 2006,
  vol. 328, pp. 117--266.

\bibitem{Welde2023}
J.~Welde, M.~D. Kvalheim, and V.~Kumar, ``The role of symmetry in constructing
  geometric flat outputs for free-flying robotic systems,'' in \emph{2023 IEEE
  International Conference on Robotics and Automation (ICRA)}, 2023, pp.
  12\,247--12\,253.

\bibitem{Brescianini2016}
D.~Brescianini and R.~D'Andrea, ``Design, modeling and control of an
  omni-directional aerial vehicle,'' in \emph{2016 IEEE International
  Conference on Robotics and Automation (ICRA)}, 2016, pp. 3261--3266.

\bibitem{Bullo1999}
F.~Bullo and R.~M. Murray, ``Tracking for fully actuated mechanical systems: a
  geometric framework,'' \emph{Automatica}, vol.~35, no.~1, pp. 17--34, 1999.

\bibitem{Lee2016}
T.~Lee, ``Optimal hybrid controls for global exponential tracking on the
  two-sphere,'' in \emph{2016 IEEE 55th Conference on Decision and Control
  (CDC)}, 2016, pp. 3331--3337.

\bibitem{Lee2010}
T.~Lee, M.~Leok, and N.~H. McClamroch, ``Geometric tracking control of a
  quadrotor {UAV} on {SE(3)},'' \emph{Proceedings of the IEEE Conference on
  Decision and Control}, pp. 5420--5425, 2010.

\bibitem{Welde2023b}
J.~Welde, M.~D. Kvalheim, and V.~Kumar, ``A compositional approach to
  certifying almost global asymptotic stability of cascade systems,''
  \emph{IEEE Control Systems Letters}, vol.~7, pp. 1969--1974, 2023.

\end{thebibliography}

\iffinal
\else

\begin{appendix}

To aid in comprehension, this appendix provides step-by-step explanations of certain technical arguments used above.
	
\subsection{Some Helpful Lemmas on Riemannian Geometry}

	\begin{lemma}
		\label{diagonal_curves_lemma}
		For a Riemmanian manifold $(Q,\kappa)$
		and a smooth ``family of curves'' ${\Gamma : \mathbb{R}^2 \to Q},$ ${(r,s) \mapsto \mathit{\Gamma}(r,s)}$, the ``diagonal'' curve given by ${\gamma : t \mapsto \mathit{\Gamma}(t,t)}$ satisfies the identity
		\begin{equation}
			\begin{aligned}
				\big(
				\nabla_{\dot{\gamma}} \dot{\gamma}
				\big)(t)
				=
				\big(
				\Diff_r \partial_r \mathit{\Gamma}
				+ 2 \Diff_r \partial_s \mathit{\Gamma}
				+ \Diff_s \partial_s \mathit{\Gamma}
				\big)(t,t).
				\label{diagonal_identity}
			\end{aligned}	
		\end{equation}
	\end{lemma}
	
	\begin{proof}
		Since this is a purely local question, we are free to work in local coordinates $(q^i)$ around any point ${\gamma(t) = \mathit{\Gamma}(r,s) \in S}$. Expressing the components of $\mathit{\Gamma}$ as $\mathit{\Gamma}(r,t) = \big(q^1(r,s),\ldots,q^n(r,s)\big)$, the components of $\gamma$ are given by $\gamma(t) = \big(q^1(t,t),\ldots,q^n(t,t)\big)$, so that
		\begin{equation}
			\dot{\gamma}(t) = \left(\tfrac{\partial q^k}{\partial r} + \tfrac{\partial q^k}{\partial s} \right) \partial_k,
		\end{equation}
		where the previous partial derivatives (and all subsequent ones) are evaluated at ${(r,s) = (t,t)}$.
		Then, by the coordinate formula for the covariant derivative along a curve, we have
		\begin{equation}
			\begin{aligned}
				\big(
				\nabla_{\dot{\gamma}} \dot{\gamma}
				\big)(t) &= 
				\Big(
				\tfrac{\partial^2 q^k}{\partial r^2} + 
				\tfrac{\partial^2 q^k}{\partial r \partial s} + 
				\tfrac{\partial^2 q^k}{\partial s \partial r} + 
				\tfrac{\partial^2 q^k}{\partial s^2}
				\\ & \hphantom{=} 
				\quad \quad \quad
				\big(\tfrac{\partial q^i}{\partial r} + \tfrac{\partial q^i}{\partial s} \big)
				\big(\tfrac{\partial q^j}{\partial r} + \tfrac{\partial q^j}{\partial s} \big)
				\Gamma^{k}_{ij}
				\Big)\partial_k,
				\label{diagonal_lemma_direct}
			\end{aligned}
		\end{equation}
		where $\Gamma^{k}_{ij}$ are the ``Christoffel symbols'' for $\nabla$.
		Meanwhile, 
		\begin{align}
			\partial_r \mathit{\Gamma} = \tfrac{\partial q^k}{\partial r} \partial_k, \quad \quad 
			\partial_s \mathit{\Gamma} = \tfrac{\partial q^k}{\partial s} \partial_k,
		\end{align}
		and using the same coordinate formula, we obtain
		\begin{align}
			\label{diagonal_lemma_r_r}
			(\Diff_r \partial_r \mathit{\Gamma})(t,t)
			&= 
			\Big(
			\tfrac{\partial^2 q^k}{\partial r^2} +
			\tfrac{\partial q^i}{\partial r} \tfrac{\partial q^j}{\partial r} 
			\Gamma^k_{ij}
			\Big) \partial_k,
			\\
			(\Diff_r \partial_s \mathit{\Gamma})(t,t)
			&=
			\Big(
			\tfrac{\partial^2 q^k}{\partial s \partial r} +
			\tfrac{\partial q^i}{\partial s} \tfrac{\partial q^j}{\partial r} 
			\Gamma^k_{ij}
			\Big) \partial_k,
			\label{diagonal_lemma_t_r}
			\\
			\label{diagonal_lemma_t_t}
			(\Diff_s \partial_s \mathit{\Gamma})(t,t)
			&= 
			\Big(
			\tfrac{\partial^2 q^k}{\partial s^2} +
			\tfrac{\partial q^i}{\partial s} \tfrac{\partial q^j}{\partial s} 
			\Gamma^k_{ij}
			\Big) \partial_k.
		\end{align}
		To complete the argument, it suffices to substitute \eqref{diagonal_lemma_r_r}-\eqref{diagonal_lemma_t_t} into \eqref{diagonal_identity} and use the equality of mixed partials and the symmetry of the Christoffel symbols to obtain \eqref{diagonal_lemma_direct}.
	\end{proof}

Note that a similar argument (found in the proof of \cite[Lemma 6.2]{RiemannianLee}) will show that ${\Diff_r \partial_s \mathit{\Gamma} = \Diff_s \partial_r \mathit{\Gamma}}$.

	\begin{lemma}
		For a homogeneous Riemannian manifold $(Q,\Phi,\kappa)$, 
		the following diagram commutes for all $g \in G$:
		\tikzcdset{every label/.append style = {font = \small}}
		\[\begin{tikzcd}
			TQ & & TQ \\
			&&\\
			{T^*Q} & & {T^*Q}
			\arrow["{\diff\Phi_g}", curve={height=-8pt}, from=1-1, to=1-3]
			\arrow["{\diff\Phi_{g^{-1}}}", curve={height=-8pt}, from=1-3, to=1-1, outer sep=2]
			\arrow["{\diff\Phi^*_{g^{-1}}}"', curve={height=8pt}, from=3-1, to=3-3, outer sep=2]
			\arrow["{\diff\Phi^*_{g}}",  curve={height=-8pt}, from=3-1, to=3-3,tail reversed, no head, ]
			\arrow["{\kappa^\sharp}", curve={height=-8pt}, from=3-1, to=1-1]
			\arrow["{\kappa^\flat}", curve={height=-8pt}, from=1-3, to=3-3]
			\arrow["{\kappa^\flat}"', curve={height=8pt}, tail reversed, no head, from=3-1, to=1-1]
			\arrow["{\kappa^\sharp}", curve={height=-8pt}, from=3-3, to=1-3]
		\end{tikzcd}\]
	\end{lemma}
	\begin{proof}
		Letting $p = \Phi_g(q)$ for some ${p, q \in Q}$, 
		then for any ${v_{p} \in T_{p}Q}$ and ${f_q \in T_q^*Q}$, we compute
		\begin{align}
			\pairing{\kappa^\flat \circ \diff \Phi_g \circ \kappa^\sharp(f_q)}{v_{p}} 
			&= 	
			\kappa\left(\diff \Phi_g \circ \kappa^\sharp(f_q),v_{p}\right)
			\\ &= \kappa\left(\kappa^\sharp(f_q), \diff \Phi_{g^{-1}} (v_{p})\right)
			\\ &= \pairing{f_q}{\diff\Phi_{g^{-1}}(v_{p})}
			\\ &= \pairing{\diff\Phi^*_{g^{-1}}(f_q)}{v_{p}}.
		\end{align}
		To complete the proof, it
		suffices to note that ${(\kappa^\sharp)^{-1} = \kappa^\flat}$, ${(\diff \Phi_g)^{-1} = \diff \Phi_{g^{-1}}}$, and ${(\diff \Phi^*_g)^{-1} = \diff \Phi^*_{g^{-1}}}$.
	\end{proof}

	\subsection{Detailed Account of Convergence in Proof of Theorem 1}
	
	This section presents a more detailed account of an argument in the end of the proof of Theorem \ref{mechanical_tracking_control_theorem}, namely that ${\dot{e}(t) \to 0_{TQ}}$ as ${t \to \infty}$ implies that ${\dist_{TQ}\big(
		\dot{q}(t),
		\dot{q}_d(t)
		\big) \to 0}$ as ${t \to \infty}$. First, we note that 
	${\dot{e}(t) \to 0_{TQ}}$ implies in particular that ${{e}(t) \to 0_{Q}}$ and thus ${\dist_Q\big({e}(t), 0_{Q}\big) \to 0}$. Since $\dist_Q$ is $\Phi$-invariant, in view of \eqref{q_and_q_d_computation} this implies that ${\dist_Q\big(q(t), q_d(t)\big) \to 0}$.  For sufficiently large $t$, it then follows from \cite[II.A.2]{EpsteinMarden2006} that we have
	\begin{equation}
		\begin{aligned}
			\dist_{TQ}\big(
			\dot{q}(t),
			\dot{q}_d(t)
			\big) 
			&\leq \\
			\dist_Q\big(q(t),&q_d(t)\big)
			+ 
			\big|\big|\tau_{\gamma_t}\big(\dot{q}(t)\big) - \dot{q}_d(t)\big|\big|_\kappa,
			\label{sasaki_inequality}
		\end{aligned}
	\end{equation}
	where ${\gamma_t : [0,1] \to Q}$ is the shortest geodesic from ${q}(t)$ to ${q}_d(t)$ and ${\tau_{\gamma_t} : T_{\gamma_t(0)}Q \to T_{\gamma_t(1)}Q}$ is parallel transport along $\gamma_t$. Then, \eqref{sasaki_inequality}, \eqref{q_dot_computation}-\eqref{q_d_dot_computation}, and the triangle inequality imply 
	\begin{align}
		\label{three_term_limit}
		&\lim_{t \to \infty}
		\dist_{TQ}\big(
		\dot{q}(t),
		\dot{q}_d(t)
		\big) \leq
		\\ &
		\nonumber
		\quad \quad \lim_{t \to \infty}
		\dist_{Q}\big(
		{e}(t),
		0_Q
		\big) +	\lim_{t \to \infty}
		\big|\big|\tau_{\gamma_t} \circ \diff \Phi_{g_d}\big(\dot{e}(t)\big)
		\big|\big|_\kappa
		\\&\quad \quad +
		\lim_{t \to \infty}
		\big|\big|
		(\tau_{\gamma_t} \circ
		\diff \Phi^{e(t)}
		- 
		\diff \Phi^{0_Q}) \big(\dot{g}_d(t)\big))
		\big|\big|_\kappa.
		\nonumber
	\end{align}
The first limit on the right-hand side has already been shown to be zero. Since parallel transport preserves Riemannian norms and $\kappa$ is $\Phi$-invariant, the second is zero as well.

We now consider the third limit. It follows from the properties of left actions (in particular, ${\Phi_g \circ \Phi _h = \Phi_{gh}}$) that ${\diff \Phi^q : TG \to TQ}$ is an equivariant map for each ${q \in Q}$ (\textit{i.e.}, ${\diff \Phi_g \circ \diff \Phi^q = \diff \Phi^q \circ \diff L_g}$). Moreover, the $\Phi$-invariance of $\kappa$ implies that ${\diff \Phi_g \circ \tau_\gamma = \tau_{(\Phi_g \circ\gamma)} \circ \diff \Phi_g}$.
Together, these facts enable us to reexpress the third limit to obtain
\begin{equation}
	\begin{aligned}
		\lim_{t \to \infty}
		&\dist_{TQ}\big(
		\dot{q}(t),
		\dot{q}_d(t)
		\big)
		\leq 
		\\& 
		\ \ \lim_{t \to \infty}
		\big|\big|
		(\tau_{\alpha_t} \circ
		\diff \Phi^{e(t)} 
		- 
		\diff \Phi^{0_Q})
		\big(\xi_d(t)\big)
		\big|\big|_\kappa,
	\end{aligned}
\end{equation}
where $\xi_d$ is the body velocity of $g_d$ and ${\alpha_t := \Phi_{g_d(t)^{-1}} \circ \gamma_t}$ can be recognized as the shortest geodesic from $e(t)$ to $0_Q$. It is now clear that for the linear map ${A_t : \mathfrak{g} \to T_{0_Q} Q, \, \xi \mapsto  \big(\tau_{\alpha_t} \circ
	\diff \Phi^{e(t)} 
	- 
	\diff \Phi^{0_Q}\big)
	(\xi)}$, we have 
\begin{align}
	&\lim_{t \to \infty}
	\dist_{TQ}\big(
	\dot{q}(t),
	\dot{q}_d(t)
	\big) \leq 
	\lim_{t \to \infty}
	\big|\big|A_t\big|\big|_{\mathrm{op}}
	\, \big|\big|\xi_d(t)\big|\big|_{\mathfrak{g}},
\end{align}
where ${||\cdot||_{\mathfrak{g}}}$ is \textit{any} norm on $\mathfrak{g}$ and ${||\cdot||_{\mathrm{op}}}$ is the operator norm induced on linear maps from $\mathfrak{g}$ to $T_{0_Q} Q$ by ${||\cdot||_{\mathfrak{g}}}$ and $\kappa$. Since $\xi_d$ is bounded for all time, uniform equivalence of norms implies the existence of a constant $C > 0$ such that 
\begin{align}
	&\lim_{t \to \infty}
	\dist_{TQ}\big(
	\dot{q}(t),
	\dot{q}_d(t)
	\big) \leq 
	C 	\lim_{t \to \infty}
	\big|\big|A_t\big|\big|_{\mathrm{op}}.
\end{align}
To complete the argument, it suffices to observe that ${\tau_{\alpha_t} \to \id}$ as ${e(t) \to 0_Q}$, and thus 
${\big|\big|A_t\big|\big|_{\mathrm{op}} \to 0}$ as ${t \to \infty}$.

\end{appendix}

\fi

\end{document}